\documentclass[10pt,a4paper]{article}
\usepackage[utf8]{inputenc}
\usepackage{geometry}
\usepackage{graphicx}
\usepackage{amssymb}
\usepackage{epstopdf}
\usepackage{mathrsfs}
\usepackage{amsmath}
\usepackage{amsthm}
\newcommand\CoS{T^*\mathbb{S}^{d-1}}

\newcommand\Point{(\omega,\eta)}

\newcommand\bel{>}
\newcommand\Tr{\mathrm{Tr}}
\newcommand\Vol{\mathrm{Vol}}
\begin{document}

\newtheorem{theoreme}{Theorem}
\newtheorem{definition}{Definition}
\newtheorem{lemme}{Lemma}
\newtheorem{remarque}{Remark}
\newtheorem{exemple}{Example}
\newtheorem{proposition}{Proposition}
\newtheorem{corolaire}{Corollary}
\newtheorem{hyp}{Hypothesis}
\newtheorem*{theo}{Theorem}

\title{Equidistribution of phase shifts in trapped scattering}
\author{Maxime Ingremeau}

\maketitle

\begin{abstract}
We consider semiclassical scattering for compactly supported
perturbations of the Laplacian and show
equidistribution of eigenvalues of the scattering matrix at
(classically) non-degenerate energy levels. The only requirement is
that sets of fixed points of certain natural scattering relations have
measure zero.
This extends the result of \cite{Equid}, where the authors proved the
equidistribution result under a similar assumption on fixed points but
with the condition that there is no trapping.
\end{abstract}

\section{Introduction}
Consider a Riemannian manifold $(X,g)$ which is Euclidean near infinity, in the sense that there exist compact sets $X_0\subset X$ and $K_0\subset \mathbb{R}^d$ such that $(X\backslash X_0,g)$ and $(\mathbb{R}^d\backslash K_0,g_{eucl})$ are isometric.

Let us consider an operator $P_h:= -h^2\Delta_g + V$, where $V\in C_c^\infty(X)$ has its support in $X_0$.
It is well-known (see for example \cite[\S 2]{Mel} or \cite[\S 3.7, \S 4.4]{Resonances}), that for any $\phi_{in}\in
C^\infty (\mathbb{S}^{d-1})$, there is a
unique solution to $\big{(}P_h-1\big{)}u=0$ satisfying, for all $x\in (X\backslash X_0)\cong (\mathbb{R}^d\backslash K_0)$:
$$u(x)=|x|^{-(d-1)/2}\big{(} e^{-i |x|/h} \phi_{in}(\omega) +
e^{i |x|/h} \phi_{out}(-\omega) \big{)} + O(|x|^{-(d+1)/2}).$$

We define the scattering matrix\footnote{which is not a matrix as soon as
$d>1$!} $S_{h}: C^\infty (\mathbb{S}^{d-1})
\longrightarrow C^\infty (\mathbb{S}^{d-1})$, which depends on $h$, by
$$S_{h}(\phi_{in}) := e^{i\pi (d-1)/2} \phi_{out}.$$
The factor $e^{i\pi(d-1)/2}$ is taken so that the scattering matrix is the identity operator when $(X,g)= (\mathbb{R}^d,g_{Eucl})$ and $V\equiv 0$.

For each $h\in(0,1]$, $S_h$ can be extended by density to a unitary operator acting on
$L^2(\mathbb{S}^{d-1})$. $S_{h}-Id$ is then a trace class operator.
 Therefore, $S_{h}$ admits a sequence of eigenvalues of modulus 1, which
converge to 1, and which we denote by $(e^{i\beta_{h,n}})_{n\in
\mathbb{N}}$.

 Our aim in this paper will be to study the behaviour of $(e^{i\beta_{h,n}})$ in the limit where $h\rightarrow 0$. To do this, we define a measure $\mu_{h}$ on $\mathbb{S}^1$ by
\[\langle \mu_{h},f \rangle := (2\pi h)^{d-1} \sum_{n\in
\mathbb{N}} f(e^{i\beta_{h,n}}),\]
for any continuous $f:\mathbb{S}^1 \longrightarrow \mathbb{C}$. This measure is not finite, but $\langle \mu_h,f\rangle$ is finite as soon as $1$ is not in the support of $f$.

Let us now state the assumptions we make on the manifold $X$ and on the potential $V$.
\paragraph{The scattering map}
We denote by $p(x,\xi)= |\xi|_g^2+V(x) : T^*X\longrightarrow \mathbb{R}$ the
classical Hamiltonian, which is the principal symbol of $P_h$. Let us write $\mathcal{E}$ for the energy layer of energy 1:
\begin{equation}\label{layer}
\mathcal{E}=\{(x,\xi)\in T^*X; ~ p(x,\xi)=1\}.
\end{equation}

We denote by $\Phi^t(\rho)$ the Hamiltonian flow for the Hamiltonian $p$.
The \emph{outgoing} and \emph{incoming} sets are defined as
$$\Gamma^\pm:=\{\rho\in \mathcal{E}, \text{ such that } \Phi^t(\rho) \text{ remains in a compact set for all } \mp t \geq 0 \}.$$
The trapped set is the compact set
$$K=\Gamma^+\cap \Gamma^-.$$

Since, away from $X_0$, the trajectories by $\Phi^t$ are just straight lines, we have that for any $\omega\in \mathbb{S}^{d-1}$, and $\eta\in \omega^\perp\subset \mathbb{R}^d$, there
exists a unique $\rho_{\omega,\eta}\in \mathcal{E}$ such that
\begin{equation}\label{defrho}
\pi_X\big{(}\Phi^t(\rho_{\omega,\eta})\big{)} = t\omega+\eta \text{   for   } t< - T_0,
\end{equation}
where $\pi_X:T^*X\rightarrow X$ denotes the projection on the base manifold $X$, and where $T_0$ is large enough, so that $K_0\subset B(0,T_0)$.
Here, $\omega$ is the \emph{incoming direction}, and $\eta$ is the \emph{impact parameter}. In the sequel, we will identify
\begin{equation*}
\{(\omega,\eta); ~\omega\in \mathbb{S}^{d-1}, \eta \in \omega^\perp\} \cong T^*\mathbb{S}^{d-1}.
\end{equation*} 

We define the \emph{interaction region} as
\begin{equation*}
\mathcal{I}:= \{(\omega,\eta)\in T^*\mathbb{S}^{d-1}; \exists t\in \mathbb{R} \text{ such that } \pi_X\big{(} \Phi^t(\rho_{\omega,\eta})\big{)}\in X_0\}.
\end{equation*}
By compactness of $X_0$, $\mathcal{I}$ is compact.

If $\rho_{\omega,\eta}\notin \Gamma^-$, then there exists $\omega'\in \mathbb{S}^{d-1}$, $\eta'\in (\omega')^\perp\subset \mathbb{R}^d$ and $t'\in \mathbb{R}$ such that for all $t\geq T_0$,
\begin{equation*}
\pi_X\big{(} \Phi^t(\rho_{\omega,\eta})\big{)} = \omega'(t-t')+\eta'.
\end{equation*}

The \emph{(classical) scattering map} is then defined as $\kappa(\omega,\eta)=(\omega',\eta')$, as represented on Figure \ref{relation}.

\begin{figure}
    \center
   \includegraphics[scale=0.6]{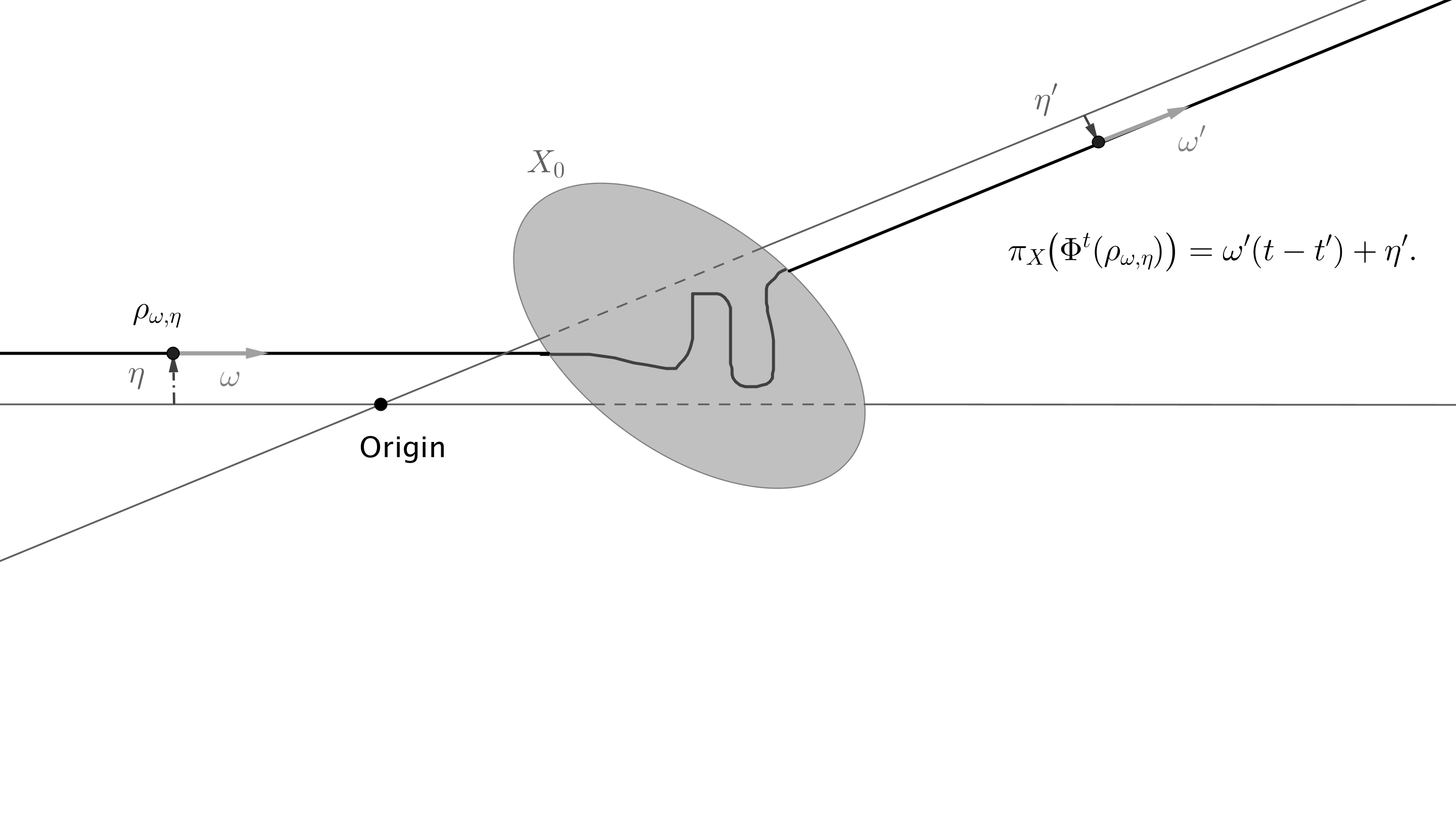}
    \caption{The scattering map $\kappa$.} \label{relation}
\end{figure}

\paragraph{The assumptions from \cite{Equid}}
The main assumption in \cite{Equid} is that
\begin{equation}\label{pakapte}
\Gamma^\pm=\emptyset.
\end{equation}

Under this assumption, $\kappa:T^*\mathbb{S}^{d-1}\rightarrow T^*\mathbb{S}^{d-1}$ is well defined. One can actually show that $\kappa$ is a symplectomorphism for the canonical symplectic structure on $T^*\mathbb{S}^{d-1}$, and in particular, it is invertible (see for example \cite{guillemin1977sojourn}). 

It is easy to see that
\begin{equation*}
\kappa(\mathcal{I})=\mathcal{I}, ~~\kappa(\mathcal{I}^c)=\mathcal{I}^c \text{  and  } \kappa|_{\mathcal{I}^c\rightarrow \mathcal{I}^c}=Id_{\mathcal{I}^c\rightarrow \mathcal{I}^c}.
\end{equation*}

The results in \cite{Equid} require a \emph{diversion hypothesis} which concerns
the periodic points of $\kappa$ in the interaction region.

For any $l\in \mathbb{Z}\backslash\{0\}$, denote the set of periodic points of $\kappa$ of period $l$ by
\begin{equation*}
\mathcal{F}_l:=\{(\omega,\eta)\in\mathcal{I};~ \kappa^l(\omega,\eta)=(\omega,\eta)\}.
\end{equation*}

The diversion hypothesis says that
\begin{equation}\label{divminko}
\forall l \in \mathbb{N}\backslash \{0\},~~\Vol\big{(}\mathcal{F}_l\big{)}=0,
\end{equation}
where $\Vol$ denotes the Liouville measure on $T^*\mathbb{S}^{d-1}$.

This hypothesis roughly says that most of the
classical trajectories in $\mathcal{E}$ which interact with the potential or the perturbation of the Euclidean metric are indeed
diverted. In \cite{Equid}, the authors work in the setting where $(X,g)\equiv (\mathbb{R}^d, g_{Eucl})$, and with $X_0=\mathrm{supp}(V)$, and they conjecture that this hypothesis holds for generic potentials.

The main result in \cite{Equid} is the following.

\begin{theo}[\cite{Equid}]
Suppose that the manifold $(X,g)$ and the potential $V$ are such that (\ref{pakapte}) and (\ref{divminko}) are satisfied.
Let $f:\mathbb{S}^1 \longrightarrow \mathbb{C}$ be a continuous function
such that $1\notin \mathrm{supp} f$. Then we have
\[ \lim \limits_{h\rightarrow 0} \langle \mu_{h},f \rangle = \frac{\Vol(\mathcal{I})}{2\pi}
\int_0^{2\pi} f(e^{i\theta}) d \theta. \]
\end{theo}

Our objective in this paper is to show that this theorem remains true if the incoming and outgoing sets are non-empty. 
We define the incoming set at infinity as
\begin{equation}\label{defcapteinfini}
 \begin{aligned}
\tilde{\Gamma}^- &:=\{ (\omega,\eta)\in T^*\mathbb{S}^{d-1} \text{ such that } \rho_{\omega,\eta}\in \Gamma^- \}.
\end{aligned}
\end{equation}

Similarly, we define the outgoing set at infinity
$\tilde{\Gamma}^+\subset T^*\mathbb{S}^{d-1}$ by:
\begin{equation}\label{defcapteinfini2}
 (\omega',\eta')\in \tilde{\Gamma}^+ \Leftrightarrow \exists (x,\xi) \in
\Gamma^+; \Phi^t(x,\xi) = t\omega'+\eta' \text{   for   } t \text{ large enough.} 
\end{equation}

Note that $\tilde{\Gamma}^\pm$ are compact subsets of $T^*\mathbb{S}^{d-1}$, since if $\eta$ is large enough, a trajectory with impact parameter $\eta$ will not meet the interaction region, and therefore cannot be trapped.

Instead of supposing (\ref{pakapte}), we will make the following assumption.

\begin{hyp}\label{Trap}
\begin{equation}\label{peukapte}
\Vol(\tilde{\Gamma}^\pm)=0.
\end{equation}
\end{hyp}

This hypothesis is very mild: as we will see in the next section, it is satisfied as soon as the energy level $\mathcal{E}$ is non-degenerate in the sense that 
\begin{equation} \label{nondegenere}
dp_{|\mathcal{E}} \neq 0.
\end{equation}

We will also make an assumption which is the analogue of (\ref{divminko}). In the case when (\ref{pakapte}) does not hold, this assumption is slightly more technical to write, since $\kappa^l$ is not well-defined on all of $\mathcal{I}$. We will therefore postpone the precise statement of this assumption to Hypothesis \ref{Div} in section \ref{Chloe}.

\paragraph{Statement of the results}
Under these hypotheses, we may state our result.
\begin{theoreme} \label{theo}

Suppose that the manifold $(X,g)$ and the potential $V$ are such that Hypotheses \ref{Trap} and \ref{Div} are satisfied.
Let $f:\mathbb{S}^1 \longrightarrow \mathbb{C}$ be a continuous function
such that $1\notin \mathrm{supp} f$. Then we have
\[ \lim \limits_{h\rightarrow 0} \langle \mu_{h},f \rangle = \frac{\Vol(\mathcal{I})}{2\pi}
\int_0^{2\pi} f(e^{i\theta}) d \theta. \]
\end{theoreme}

\begin{remarque}
For simplicity, we shall only state and prove this result for smooth potentials, but it should still be true for less regular potentials, as long as the Hamiltonian dynamics is well defined. The proof should work without many changes for a potential $V\in C^1_c(X;\mathbb{R})$.
\end{remarque}

As in \cite{Equid}, we may deduce the following corollary.
\begin{corolaire}\label{repart}
Let $0<\phi_1<\phi_2< 2\pi$ be angles, and let $N_h(\phi_1,\phi_2)$ be the number of eigenvalues $e^{i\beta_{h,n}}$ of $S_h$ with $\phi_1\leq \beta_{h,n}\leq \phi_2$ modulo $2\pi$. Then we have
\begin{equation*}
\lim\limits_{h\rightarrow 0} (2\pi h)^{d-1} N_h(\phi_1,\phi_2) = \Vol(\mathcal{I}) \frac{\phi_2-\phi_1}{2\pi}.
\end{equation*}
\end{corolaire}
The proof of Corollary \ref{repart} is exactly the same as that of Corollary 1.2 in \cite{Equid}: we simply approach uniformly the indicator function $\mathbf{1}_{[\phi_1,\phi_2]}$ by continuous functions and use Theorem \ref{theo}. We refer to \cite{Equid} for more details.

\paragraph{Relation to other works}
The distribution of the eigenvalues of the scattering matrix has been studied since the eighties (\cite{birman1982asymptotic}, \cite{birman1984asymptotic}, \cite{sobolev1985phase}). More recently, in the (non-semiclassical) high-energy limit, it was studied in \cite{EquidBP}, and extended to more general Hamiltonians in \cite{EquidMag} and \cite{EquidNakamura}.
For related topics in the physics literature for obstacle scattering, see \cite{EquidObst}.

In the semiclassical setting, equidistribution of phase shifts was first observed in \cite{EquidSph} for spherically symmetric potentials, and in \cite{Equid} for more general non-trapping potentials. It was also studied in \cite{EquidPol} for long-range potentials, without any assumption on the classical dynamics. In \cite{zelditch1999spacing}, the authors obtain much finer results on the distribution of phase shifts in the semiclassical limit for a family of surfaces of revolution.

Just as in \cite{Equid}, the main tool in the proof of the equidistribution of phase shifts is the fact that the scattering matrix is a Fourier Integral Operator associated to the scattering map microlocally away from the incoming and outgoing directions. This was proven in \cite{Alex}, and also in \cite{HaWu} in a geometric non-trapping setting.

The scattering map is trivial outside of the interaction region, while it can be very complicated inside the interaction region. This mixed behaviour is somehow similar to the situation described in \cite{MarOK}, where the authors prove a Weyl law for general systems for which the phase space can be separated into a part where the classical dynamics is periodic, and another where its is ergodic.

\paragraph{Acknowledgements}
The author would like to thank Stéphane Nonnenmacher for supervising this project and for useful discussions. He would also like to thank Jesse Gell-Redman, Andrew Hassell and Steve Zelditch for their comments which helped improve the first version of the manuscript. Last but not least, the author would like to thank the anonymous referee who helped improve many aspects of this paper, in particular by suggesting Lemma \ref{conditionnondegener} and some of the estimates in section \ref{fin}.

The author is partially supported by the Agence Nationale de la Recherche project GeRaSic (ANR-13-BS01-0007-01).
\section{Classical dynamics} \label{Chloe}

Recall that, if $\Point\in \CoS$, $\rho_{\omega,\eta}$ was defined in (\ref{defrho}), and that we defined the sets $\tilde{\Gamma}^\pm$ in (\ref{defcapteinfini}) and (\ref{defcapteinfini2}). 

Although we will not use it in the sequel, let us prove now the fact announced in the introduction that (\ref{nondegenere}) implies Hypothesis \ref{Trap}. The proof is standard (it is very similar to that of \cite[Proposition 6.5]{Resonances} or \cite[Proposition A.3]{gerard1987semiclassical}), but we recall it for the reader's convenience.
\begin{lemme}\label{conditionnondegener}
Suppose that $p: T^*X\rightarrow \mathbb{R}$ is such that $dp_{|\mathcal{E}} \neq 0$. Then we have $\Vol(\tilde{\Gamma}^\pm)=0.$
\end{lemme}
\begin{proof}
Suppose that $dp_{|\mathcal{E}} \neq 0$. $\mathcal{E}$ is then a smooth manifold, which can be equipped with the Liouville measure $\mu$. This measure is invariant by the Hamiltonian flow $(\Phi^t)$.

Note that, outside of $X_0$, this measure is just the Lebesgue measure on $S^*(X\backslash X_0)$, 
so that, if we define for $r_0$ large enough and $r_1>r_0$ the annulus $$C_{r_0,r_1} := (B(0,r_1)\backslash B(0,r_0))
\subset (\mathbb{R}^d\backslash K_0) \simeq (X\backslash X_0),$$ we have
\begin{equation*}
\Vol (\tilde{\Gamma}^\pm)=0 ~~\Leftrightarrow ~~ \exists ~0< r_0< r_1~~ \text{large enough such that } \mu \big{(} \Gamma^\pm \cap S^*C_{r_0,r_1}\big{)}=0.
\end{equation*}

Suppose for contradiction that we may find $0< r_0< r_1$ such that $\mu \big{(} \Gamma^\pm \cap S^*C_{r_0,r_1}\big{)}>0.$ 

Since the motion of a point in $\Gamma^\pm\cap S^*C_{r_0,r_1}$ as $\pm t\geq 0$  is just a straight line, we may find a time $t_0= t_0(r_0,r_1)$ such that for any $j\geq 1$, $\Phi^{\pm j t_0} \big{(} \Gamma^\pm \cap S^*C_{r_0,r_1}\big{)} \cap \big{(} \Gamma^\pm \cap S^*C_{r_0,r_1}\big{)} = \emptyset$. Since $\Phi^{t_0}$ is a diffeomorphism, we then have that for all $j,j'\in \mathbb{N}$ with $j\neq j',$ $$\Big{(}\Phi^{\mp j t_0} \big{(} \Gamma^\pm \cap S^*C_{r_0,r_1}\big{)}\Big{)} \cap \Big{(}\Phi^{\mp j' t_0} \big{(} \Gamma^\pm \cap S^*C_{r_0,r_1}\big{)}\Big{)}  = \emptyset.$$

Since $\mu$ is invariant by the Hamiltonian flow, we have that
\begin{equation*}
\begin{aligned}
\mu \Big{(} \bigcup_{j=0}^\infty \Phi^{\mp j t_0} \big{(} \Gamma^\pm \cap S^*C_{r_0,r_1}\big{)} \Big{)}
&= \sum_{j=0}^\infty \mu\Big{(}\Phi^{\mp j t_0} \big{(} \Gamma^\pm \cap S^*C_{r_0,r_1}\big{)}\Big{)}\\
&= \sum_{j=0}^\infty \mu \big{(} \Gamma^\pm \cap S^*C_{r_0,r_1}\big{)}\\
&= +\infty,
\end{aligned}
\end{equation*}
by assumption. But for all $j\geq 0$, $\Phi^{\mp j t_0} \big{(} \Gamma^\pm \cap S^*C_{r_0,r_1}\big{)}$ belongs to a compact region of $\mathcal{E}$, where the base points are either in $X_0$, or in $B(0,r_1)\subset \mathbb{R}^d$. Hence, we must have $\mu \Big{(} \bigcup_{j=0}^\infty \Phi^{\mp j t_0} \big{(} \Gamma^\pm \cap S^*C_{r_0,r_1}\big{)} \Big{)} < + \infty$, a contradiction.
\end{proof}

If $\rho_{\omega,\eta}\notin \Gamma^-$, then there exists $\omega'\in \mathbb{S}^{d-1}$, $\eta'\in (\omega')^\perp\subset \mathbb{R}^d$ and $t'\in \mathbb{R}$ such that for all $t$ large enough,
\begin{equation*}
\pi_X\big{(} \Phi^t(\rho_{\omega,\eta})\big{)} = \omega'(t-t')+\eta'.
\end{equation*}

We may then define the (classical) scattering map
$$\kappa : T^*\mathbb{S}^{d-1} \backslash \tilde{\Gamma}^-
\longrightarrow T^*\mathbb{S}^{d-1} \backslash \tilde{\Gamma}^+ $$
by $(\omega',\eta') = \kappa (\omega,\eta)$. $\kappa$ is then a
symplectomorphism.

We define the ``good" sets $(\mathcal{G}^+_k)\subset \CoS$ and $(\mathcal{G}^-_k)\subset
\CoS$ by induction for $k\in \mathbb{N}$, by
\begin{equation}\label{defbonensemble}
\begin{aligned}
\mathcal{G}^+_0 &:= \CoS \backslash \tilde{\Gamma}^-,~~~~ &\mathcal{G}^+_{k+1} &:=
\{\Point \in \mathcal{G}^+_k \text{ such that } \kappa\Point \in \mathcal{G}^+_0\}\\
\mathcal{G}^-_0 &:= \CoS \backslash \tilde{\Gamma}^+,~~~~ &\mathcal{G}^-_{k+1} &:=
\{\Point \in \mathcal{G}^-_k \text{ such that } \kappa^{-1}\Point \in \mathcal{G}^-_0\}.
\end{aligned}
\end{equation}
The scattering map may then be iterated and inverted, to obtain for
any $k\geq 1$  symplectomorphisms
$$\kappa^k : \mathcal{G}^+_{k-1} \longrightarrow \mathcal{G}^-_{k-1},$$
$$\kappa^{-k} : \mathcal{G}^-_{k-1} \longrightarrow \mathcal{G}^+_{k-1},$$
or, written in a more condensed way, we may define for $k\in \mathbb{Z}\backslash \{0\}$, $\kappa^k: \mathcal{G}_{|k|-1}^{\epsilon(k)} \rightarrow \mathcal{G}_{|k|-1}^{-\epsilon(k)}$, where $\epsilon(k)$ is the sign of $k$.

We also define, for $k\in \mathbb{Z}\backslash \{0\}$ 
\begin{equation}\label{bad}
\mathcal{B}_k:= T^*\mathbb{S}^{d-1}\backslash \mathcal{G}_{|k|-1}^{\epsilon(k)}.
\end{equation}
$\mathcal{B}_k$ is hence the ``bad" set where $\kappa^k$ is not well-defined.

\begin{lemme}\label{mesnul}
Suppose Hypothesis \ref{Trap} is satisfied, and let $k\in \mathbb{Z}\backslash\{0\}$. Then $\mathcal{B}_k$ has zero Liouville measure.
\end{lemme}
\begin{proof}
By assumption, $\tilde{\Gamma}^\pm$ has zero Liouville measure. Since $\kappa$ preserves the Liouville measure, we see from (\ref{defbonensemble}) that $\mathcal{G}_k^\pm$ has full measure. 
\end{proof}

For $l\in
\mathbb{Z}\backslash\{0\}$, we define the set of $l$-periodic interacting points as
\begin{equation}\label{defperiod}
\mathcal{F}_l:= \{ \Point\in \mathcal{I}\cap
\mathcal{G}^{\epsilon(l)}_{|l|-1}; \kappa^l\Point=\Point \},
\end{equation}
where $\epsilon(l)$ is the sign of $l$. Note that this set is closed.

Our \emph{diversion hypothesis} is the following.
\begin{hyp}\label{Div}
For any $l\in \mathbb{Z}\backslash\{0\}$, the Liouville measure of $\mathcal{F}_l$ is 0.
\end{hyp}

We conjecture that if $V\equiv 0$ and if $(\mathrm{Int}(X_0),g)$ has (not uniformly) negative curvature, where $\mathrm{Int}(X_0)$ denotes the interior of $X_0$ then this hypothesis holds.

Note that, since $\kappa^l$ preserves the volume, this Hypothesis is equivalent to the seemingly weaker statement that
for any $l\in \mathbb{N}\backslash\{0\}$, the Liouville measure of $\mathcal{F}_l$ is 0.

Note also that this hypothesis implies that 
\begin{equation}\label{lordship}
\Vol (\partial \mathcal{I})= 0.
\end{equation}
Indeed, a point in the boundary of $\mathcal{I}$ is in $\mathcal{I}$ because $\mathcal{I}$ is closed, and it is fixed by $\kappa$.

Before proving Theorem \ref{theo}, we need to recall a few facts and definitions from semiclassical analysis.

\section{Refresher on semiclassical analysis}
\subsection{Pseudodifferential calculus} \label{greve}
Let $Y$ be a compact manifold ($Y$ will often be $\mathbb{S}^{d-1}$ in the sequel).
We shall say that a function $a(x,\xi;h)\in C^{\infty}(T^*Y\times (0,1])$ is in the class $S^{comp}(T^*Y)$ if it can be written as 
\begin{equation*}
a(x,\xi;h)= \tilde{a}_h(x,\xi) + O\Big{(}\Big{(}\frac{h}{\langle \xi \rangle}\Big{)}^\infty \Big{)},
\end{equation*}
where $\tilde{a}_h\in C_c^\infty(T^*Y)$, with $\mathrm{supp} (\tilde{a}_h)\subset \Omega$ for some bounded  open set $\Omega$ independent of $h$, and where $\tilde{a}_h$ is bounded in any $C^k(\Omega)$ norm independently of $h$.

We associate to $S^{comp}(T^*Y)$ the algebra of pseudodifferential operators
$\Psi_h^{comp}(Y)$, through a surjective quantization map
\begin{equation*}Op_h:S^{comp}(T^*Y)\longrightarrow \Psi^{comp}_h(Y).
\end{equation*} This quantization
map is defined using coordinate charts, and the standard Weyl quantization
on $\mathbb{R}^d$. It is therefore not intrinsic. However, the principal
symbol map
\begin{equation*}\sigma_h : \Psi^{comp}_h (Y)\longrightarrow S^{comp}(T^*Y)/
h S^{comp}(T^*Y)
\end{equation*} is intrinsic, and we have
\begin{equation*}\sigma_h(A\circ B) = \sigma_h (A) \sigma_h(B)
\end{equation*}
and
\begin{equation*}\sigma_h\circ Op: S^{comp}(T^* Y) \longrightarrow S^{comp} (T^*Y) /h
S^{comp}(T^*Y)
\end{equation*}
is the natural projection map.

For more details on all these maps and their construction, we refer the reader
to \cite[Chapter 14]{Zworski_2012}.

For $a\in S^{comp}(T^*Y)$, we say its essential support is equal to a given
compact $\mathcal{K}\Subset T^*Y$,
\begin{equation*} \text{ ess supp}_h a = \mathcal{K} \Subset T^*Y,
\end{equation*}
if and only if, for all $\chi \in C_c^\infty(T^*Y)$,
\begin{equation*}\mathrm{supp} (\chi) \subset (T^*Y\backslash K) \Rightarrow \chi a \in h^\infty S^{comp}(T^*
Y).
\end{equation*}
For $A\in \Psi^{comp}_h(Y), A=Op_h(a)$, we define the wave front set of $A$ as:
\begin{equation*}WF_h(A)= \text{ ess supp}_h a,
\end{equation*}
noting that this definition does not depend on the choice of the
quantization.

\subsection{Lagrangian states and Fourier Integral Operators}\label{DSK}
In this section, we will recall the definition of Fourier Integral Operators with notations inspired by \cite{DG}. We refer to this paper and to the references therein for the classical proofs we omit.
\paragraph{Phase functions}
Let $\phi(y,\theta)$ be a smooth real-valued function on some open subset $U_\phi$ of $Y\times \mathbb{R}^L$, for some $L\in \mathbb{N}$. We call $x$ the \emph{base variables} and $\theta$ the \emph{oscillatory variables}. We say that $\phi$ is a \emph{nondegenerate phase function} if the differentials $d(\partial_{\theta_1} \phi)...d(\partial_{\theta_L}\phi)$ are linearly independent on the \emph{critical set }
\begin{equation*}
C_\phi:=\{ (y,\theta); \partial_\theta \phi =0 \} \subset U_\phi.
\end{equation*}
In this case
\begin{equation*}
\Lambda_\phi:= \{(y,\partial_y \phi(y,\theta)); (y,\theta)\in C_\phi \} \subset T^*Y
\end{equation*}
is an immersed Lagrangian manifold. By shrinking the domain of $\phi$, we can make it an embedded Lagrangian manifold. We say that $\phi$ \emph{generates} $\Lambda_\phi$.

\paragraph{Lagrangian states}
Given a phase function $\phi$ and a symbol $a\in S^{comp}(U_\phi)$, consider the $h$-dependent family of functions
\begin{equation}\label{massai}
u(y;h)= h^{-L/2} \int_{\mathbb{R}^L} e^{i\phi(y,\theta)/h} a(y,\theta;h) d\theta.
\end{equation}
We call $u=(u(h))$ a \emph{Lagrangian state}, (or a \emph{Lagrangian distribution}) generated by $\phi$. 

\begin{definition}\label{Grenoble}
Let $\Lambda\subset T^*Y$ be an embedded Lagrangian submanifold. We say that an $h$-dependent family of functions $u(y;h)\in C_c^\infty(Y)$ is a (compactly supported and compactly microlocalized) \emph{Lagrangian state associated to $\Lambda$}, if it can be written as a sum of finitely many functions of the form (\ref{massai}), for different phase functions $\phi$ parametrizing open subsets of $\Lambda$, plus an $O(h^\infty)$ remainder in the $C^\infty(Y)$ topology. We will denote by $I^{comp}(\Lambda)$ the space of all such functions.
\end{definition}

\paragraph{Fourier integral operators}
Let $Y, Y'$ be two manifolds of the same dimension $d$, and let $\kappa$ be a symplectomorphism from an open subset of $T^*Y$ to an open subset of $T^*Y'$. Consider the Lagrangian
\begin{equation*}
\Lambda_\kappa =\{(y,\nu;y',-\nu'); \kappa(y,\nu)=(y',\nu')\}\subset T^*Y\times T^*Y'= T^*(Y\times Y').
\end{equation*}
A compactly supported operator $T:\mathcal{D}'(Y)\rightarrow C_c^\infty(Y')$ is called a (semiclassical) \emph{Fourier integral operator} associated to $\kappa$ if its Schwartz kernel $K_T(y,y')$ lies in $h^{-d/2}I^{comp}(\Lambda_\kappa)$. We write $T\in I^{comp}(\kappa)$. Note that such an operator is automatically trace class. The $h^{-d/2}$ factor is explained as follows: the normalization for Lagrangian states is chosen so that $\|u\|_{L^2}\asymp 1$, while the normalization for Fourier integral operators is chosen so that $\|T\|_{L^2(Y)\rightarrow L^2(Y')} \asymp 1$.

Note that if $\kappa\circ \kappa'$ is well defined, and if $T\in I^{comp}(\kappa)$ and $T'\in I^{comp}(\kappa')$, then $T\circ T'\in I^{comp} (\kappa\circ \kappa')$.

The main property we will use about FIOs is the following, which is an easy version of \cite[Proposition 2]{Equid}.

\begin{lemme}\label{pafix}
Let $\kappa : T^*Y \supset U \rightarrow V\subset T^*Y$ have no fixed point, and let $T\in I^{comp}(\kappa)$. Then
\begin{equation*}
\Tr (T) = O(h^\infty)
\end{equation*}
\end{lemme}
\begin{proof}(Sketch)
By definition, the integral kernel of $T$ can be written as a finite sum of terms of the form
\begin{equation*}
(2\pi h)^{-L}\int_{\mathbb{R}^L} e^{i\phi(y,y';\theta)/h}a(y,y',\theta,h)d\theta,
\end{equation*}
where $\phi$ locally parametrises $\Lambda_\kappa$ in the sense that in some open subset $U\subset T^*(Y\times Y')$, we have
\begin{equation*}\Lambda_\kappa \cap U=\{(y,\partial \phi_{y'} (y,y',\theta), y', -\partial_{y} \phi(y,y',\theta)); (y,y',\theta) \text{ such that } \partial_\theta \phi(y,y',\theta)=0\}.
\end{equation*}

The trace is then given by a sum of terms of the form
\begin{equation*}
\frac{1}{(2\pi h)^{L+d-1}} \int_Y \int_{\mathbb{R}^L} e^{i\frac{\phi(y,y;\theta)}{h}} a(y,y,\theta,h) d\theta dy.
\end{equation*}

The fact that $\kappa$ has no fixed point implies that if $(y,y,\theta)$ are such that $\partial_\theta \phi(y,y,\theta)=0$, we have $\partial_y\big{[}\phi(y,y,\theta)\big{]}=\big{[}\partial_y\phi(y,y',\theta)+ \partial_{y'}\phi(y,y',\theta)\big{]}_{y=y'}\neq 0$. Then, by non-stationary phase, we obtain the result.
\end{proof}

\subsection{The scattering matrix as a FIO}\label{FIO}
The main result we will use about the scattering matrix in this paper is \cite[Theorem 5]{Alex}, which can be rephrased as follows.

\begin{theo}[Alexandrova 2005]
(i) Let $\Point\in \mathcal{G}^+_0$. If $U$ is an open neighbourhood of $\Point$ contained in $\mathcal{G}^+_0$ and $A\in \Psi_h^{comp}(\mathbb{S}^{d-1})$ is such that $WF_h(A)\subset U$, then we have $S_h A\in I^{comp}(\kappa|_U)$.

(ii) $S_h$ is microlocally equal to the identity away from the interaction region in the following sense. If $a\in S^{comp}(\mathbb{S}^{d-1})$ is such that $a\equiv 1$ near $\mathcal{I}$, then we have 
\begin{equation}\label{fleur2}
\|(S_h-Id)(Id-Op_h(a))\|_{L^2(\mathbb{S}^{d-1})\rightarrow L^2(\mathbb{S}^{d-1})} = O(h^\infty).
\end{equation}
\end{theo}

\section{Trace formula}
Our aim in this section will be to prove the following proposition, which
is the cornerstone of the proof in \cite{Equid}.
\begin{proposition} \label{chimp}
Suppose that the manifold $(X,g)$ and the potential $V$ are such that Hypotheses \ref{Trap} and \ref{Div} are satisfied, and let $k\in \mathbb{Z}\backslash\{0\}$.
Then we have
\begin{equation}\label{hello}
\Tr \big{(}S_{h}^k-Id\big{)} = -\frac{\Vol(\mathcal{I})}{(2\pi h)^{d-1}} +o(h^{-(d-1)}).
\end{equation}
\end{proposition}
\begin{proof}
To prove this proposition, we fix $k\in \mathbb{Z}\backslash\{0\}$, and build an adapted partition of
unity. 

\paragraph{Partition of unity}
Recall that $\mathcal{B}_k$ was defined in (\ref{bad}), and is the set where $\kappa^k$ is not well-defined.
We will write
\begin{equation*}
\mathcal{P}_k:= \mathcal{B}_k\cup \mathcal{F}_k\subset T^*\mathbb{S}^{d-1},
\end{equation*}
where $\mathcal{F}_k$ are as in (\ref{defperiod})
This set is closed, has zero Liouville measure by Lemma \ref{mesnul} and Hypothesis \ref{Div}, and the map $\kappa^k$ is well-defined and has no fixed points in $\mathcal{I}\backslash \mathcal{P}_k$.

Since $\mathcal{P}_k$ is closed with zero Liouville measure, by outer regularity of the Liouville measure, we may find for each $\varepsilon\bel 0$ a cut-off function $\chi_\varepsilon^k\in C_c^\infty(T^*\mathbb{S}^{d-1};[0,1])$ such that $\chi_\varepsilon^k\Point=1$ if $\Point\in \mathcal{P}_k$, such that the support of $\chi_\varepsilon^k$ is contained in an $\varepsilon$-neighbourhood of $\mathcal{I}$, and such that the Liouville measure of the support of $\chi_\varepsilon^k$ is smaller than $\varepsilon$:
\begin{equation*}
\Vol(\mathrm{supp}(\chi_\varepsilon^k))\leq \varepsilon.
\end{equation*}
We denote by $Op_h(\chi_\varepsilon^k)$ the Weyl quantization of $\chi_\varepsilon^k$, as defined in section \ref{greve}.

We also take $\psi_\varepsilon^1\in C_c^\infty(\CoS;[0,1])$ such that $\psi_\varepsilon^1=1$ near $\mathcal{I}$ and $\psi_\varepsilon^1\Point=0$ if $d(\Point, \mathcal{I})\geq \varepsilon$ and $\psi_\varepsilon^2\in C_c^\infty(\CoS;[0,1])$ such that $\psi_\varepsilon^2=0$ outside of $\mathcal{I}$, and $\psi_{\varepsilon}^2=1$ outside of an $\varepsilon$-neighbourhood of $T^*\mathbb{S}^{d-1}\backslash \mathcal{I}$ (see Figure \ref{cutoff}).
\begin{figure}
    \center
   \includegraphics[scale=0.4]{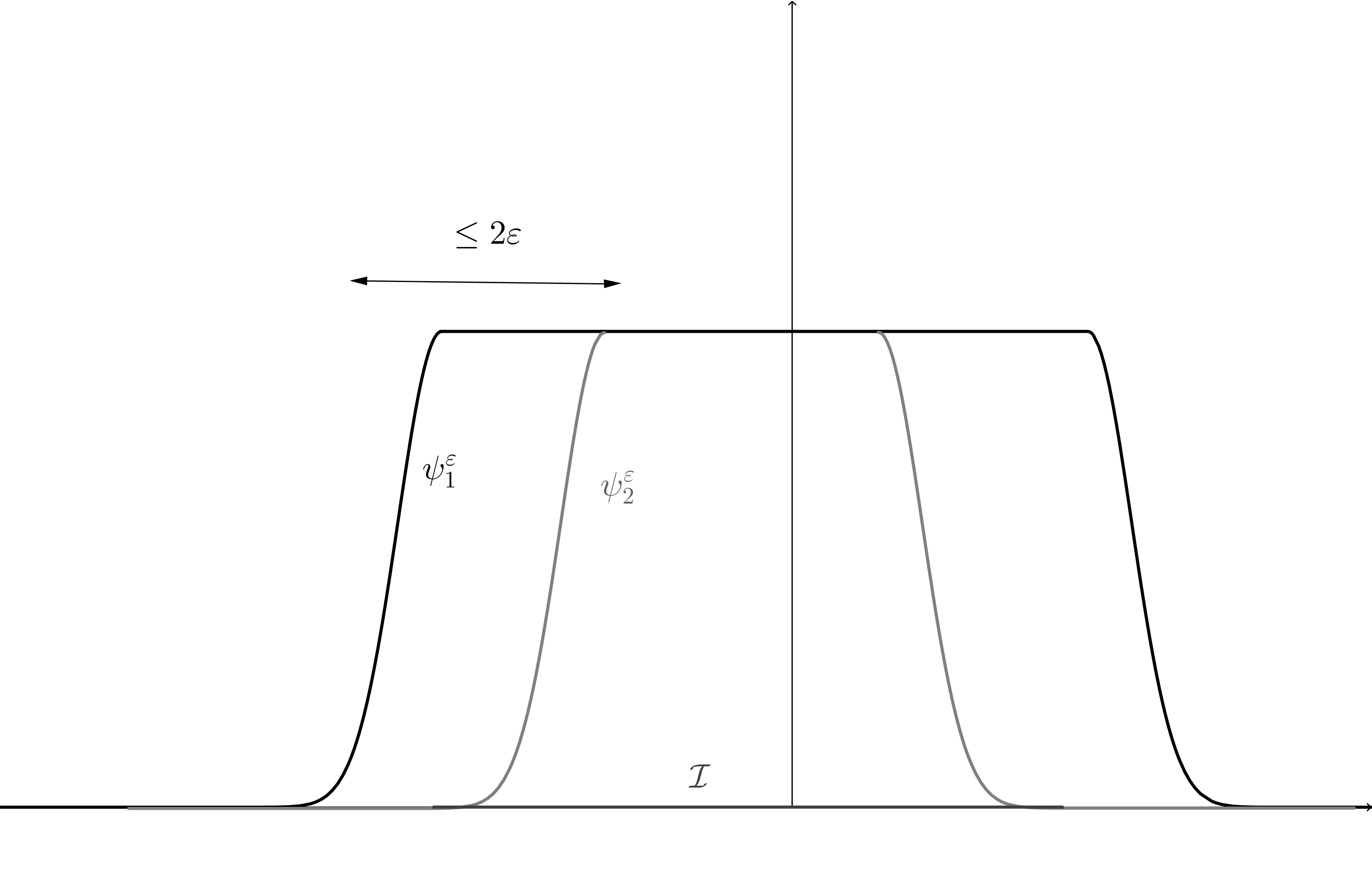}
    \caption{The cut-off functions $\psi_1^\varepsilon$ and $\psi_2^\varepsilon$} \label{cutoff}
\end{figure}

Note that we have for all $\Point\in\CoS$, $\psi_\varepsilon^1\Point \geq \psi_\varepsilon^2\Point$, and that $\|\psi_\varepsilon^1-\psi_\varepsilon^2\|_{L^1}=O(\epsilon)$ thanks to (\ref{lordship}).

We have
\begin{equation}\label{partition}
1 =  (1-\psi_\varepsilon^1)+ \psi_\varepsilon^2(1-\chi_\varepsilon^k)+ \psi_\varepsilon^2\chi_\varepsilon^k  + (\psi_\varepsilon^1-\psi_\varepsilon^2).
\end{equation}

The first term corresponds to points outside of the interaction region. The second term corresponds to points in the interaction region which are neither trapped nor fixed, while the last two terms have a support of a size $O(\varepsilon)$. We shall compute the trace of $(S_h^k-Id)$ using this decomposition.

\paragraph{Trace inside the interaction region}
By Alexandrova's Theorem (see Section \ref{FIO}), we have that $S_h^{k} Op_h(\psi_\varepsilon^2(1-\chi_\epsilon^k))$ is a Fourier integral operator associated to $\kappa^{k}_{|_{\mathcal{I}\backslash \mathcal{P}_k}}$ microlocally near 
$(\mathcal{I}\backslash \kappa^{k}(\mathcal{P}_k))\times (\mathcal{I}\backslash \mathcal{P}_k)$.

Since, by definition of $\mathcal{P}_k$, $\kappa^{k}$ has no fixed points in $\mathcal{I}\backslash \mathcal{P}_k$, Lemma \ref{pafix} tells us that
\begin{equation*}
\Tr (S_h^{k} Op_h(\psi_\varepsilon^2(1-\chi_\epsilon^k)) = O(h^\infty).
\end{equation*}

This implies that
\begin{equation}\label{inside}
\begin{aligned}
\Tr\big{(}(S_{h}^k-Id)Op_h(\psi_\varepsilon^2(1-\chi_\epsilon^k))\big{)} &= \Tr (Op_h(\psi_\varepsilon^2(1-\chi_\epsilon^k))) + O(h^\infty)\\
&= \frac{1}{(2\pi h)^{d-1}} \int_{\CoS} \psi_\varepsilon^2(1-\chi_\varepsilon^k) +O(h^{2-d})\\
&=\frac{1}{(2\pi h)^{d-1}} \Vol(\mathcal{I}) +h^{-(d-1)}r_\varepsilon +O(h^{2-d}),
\end{aligned}
\end{equation}
where $r_\varepsilon$ is independent of $h$, and is a $O(\varepsilon)$. To go from the first line to the second, we used the standard formula of the trace of a pseudodifferential operator as the integral of its symbol (see \cite[Appendix C]{Zworski_2012}).

\paragraph{Trace outside of the interaction region}
To estimate the trace outside of the interaction region, we shall consider an orthonormal basis of $L^2(\mathbb{S}^{d-1})$ made of spherical harmonics $\phi^m_\ell$ satisfying $(\Delta_{\mathbb{S}^{d-1}}-\ell(\ell+d-1))\phi^m_\ell=0$, where $\ell\in \mathbb{N}$, $0\leq m\leq d_\ell$. Here $d_\ell=O(\ell^{d-2})$, as can be seen using Weyl's law.

Let $R\bel 0$ be large enough so that 
\begin{equation*}
\mathcal{I}\subset \{\Point\in \CoS; |\eta|\leq R\}.
\end{equation*}

We need the following elementary lemma:
\begin{lemme}\label{harmononstat}
For all $R'\bel R$, $x\in B(0,R)\subset \mathbb{R}^d$, $h\bel 0$ and all $\ell\geq R'/h$, $m\leq d_\ell$, we have
\begin{equation*}
\int_{\mathbb{S}^{d-1}}  e^{i \langle\omega,x\rangle/h}\phi^m_\ell(\omega)d\omega = O \Big{(}\Big{(} \frac{R}{h \ell}\Big{)}^\infty\Big{)}.
\end{equation*}
\end{lemme}
\begin{proof}
We have, for any $n\in \mathbb{N}$, by integration by parts,
\begin{equation*}
\begin{aligned}
\int_{\mathbb{S}^{d-1}}  e^{i \langle\omega,x\rangle/h}\phi_\ell^m(\omega) d\omega= \frac{1}{(\ell(\ell+1))^n} \int_{\mathbb{S}^{d-1}} \phi^m_\ell(\omega) \Delta^n e^{i \langle\omega,x\rangle/h} d\omega.
\end{aligned}
\end{equation*}

Now, $\Delta^n e^{i \langle\omega,x\rangle/h}$ is bounded by $\Big{(}\frac{|x|}{h}\Big{)}^{2n}$ times a polynomial which depends only on $n$. The result follows.
\end{proof}

The following lemma allows us to estimate the trace outside of the interaction region.

\begin{lemme} \label{Out}
Suppose Hypotheses \ref{Trap} and \ref{Chloe} are satisfied, and take $k\in \mathbb{Z}$. We have
$$\Tr\big{(}(S_{h}^k-Id)(Id-Op_h(\psi_\varepsilon^1)) \big{)} = O(h^\infty).$$
\end{lemme}

\begin{proof}
Let us note first that thanks to (\ref{fleur2}), for each $\varepsilon>0$, $\ell\in \mathbb{N}$ and $m=1,...,d_\ell$ , we have 
\begin{equation}\label{dehors}
\|(S_h^k-Id)(Id-Op_h(\psi_\varepsilon^1))\phi_\ell^m\|=O(h^\infty).
\end{equation}

We have
\begin{equation*}
\begin{aligned}
\mathrm{Tr} \big{(}(S_{h}^k-Id)(Id-Op_h(\psi_\varepsilon^1))\big{)} &=
\sum_{\ell \in \mathbb{N}} \sum_{m=1}^{d_\ell}
\langle \phi^m_\ell, (S_h^k-Id)(Id-Op_h(\psi_\varepsilon^1)) \phi^m_\ell \rangle\\
&= \sum_{\ell <R'/h^2} \sum_{m=1}^{d_\ell}
\langle \phi^m_\ell, (S_h^k-Id)(Id-Op_h(\psi_\varepsilon^1)) \phi^m_\ell \rangle \\
&+ \sum_{\ell \geq R'/h^2} \sum_{m=1}^{d_\ell}
\langle \phi^m_\ell, (S_h^k-Id)(Id-Op_h(\psi_\varepsilon^1)) \phi^m_\ell \rangle\\
&= \sum_{\ell \geq R'/h^2} \sum_{m=1}^{d_\ell}
\langle \phi^m_\ell, (S_h^k-Id) \phi^m_\ell \rangle +O(h^\infty),
\end{aligned}
\end{equation*}
where $R'>R+\varepsilon$. Here, we dealt with the sum for $\ell<R'/h^2$ using (\ref{dehors}) and the fact that $d_\ell=O(\ell^{d-2})$.

Let us now bound the sum for $\ell \geq R'/h^2$.
Let us denote by $a_k(\omega,\omega';h)$ the integral kernel of $S_h^k-Id$. Recall the following representation\footnote{Note that this expression for $a(\omega,\omega';h)$ is smooth (and even analytic) in $\omega$ and $\omega'$, which shows that $S_h-Id$ is trace-class.} for $a_1$, which can be found in \cite{Alex}, equation (59):
\begin{equation}\label{ampliscatt}
a_1(\omega,\omega';h)=c(d,h)\int_{\mathbb{R}^d} e^{i \langle\omega,x\rangle/h} \big{(}[h^2\Delta,\chi_2]R_h [h^2\Delta,\chi_1] e^{i\langle \omega',\cdot\rangle/h}\big{)}(x)dx,
\end{equation}
where $R_h=(P_h-(1+i0))^{-1}$ is the outgoing resolvent, and $\chi_1$, $\chi_2$ are some functions in $C_c^\infty(X)$. Here, $c(d,h):=e^{-i\pi(d-3)/4} 2^{(-d+9)/4} (\pi h)^{(-d+1)/2}$ is a constant which depends polynomially in $h$. It was proven in \cite[\S 2]{petkov2001semi} that the representation (\ref{ampliscatt}) is indeed independent of the choice of the cut-off functions $\chi_1$ and $\chi_2$.

Now, from \cite[Theorem 4]{BurqRes} (see also \cite{cardoso2001uniform} for a more general statement, and \cite{datchev2014quantitative}, \cite{shapiro2016semiclassical} for similar statements with less regularity assumptions on $V$), we have that if $r_1\bel 0$ is large enough, and if $r_2\bel r_1$, then
\begin{equation}\label{burq}
\|\mathrm{1}_{r_1\leq |x|\leq r_2} R_h \mathrm{1}_{r_1\leq |x|\leq r_2}\|_{L^2(\mathbb{R}^d)\rightarrow L^2(\mathbb{R}^d)}= O(h^{-1}).
\end{equation}

From this, we get that $a_1(\omega,\omega';h) = \int_{\mathbb{R}^d} e^{i \langle\omega,x\rangle/h} f_1(x,\omega';h) dx$, where $f_1$ is a function which is smooth in $x$ and $\omega'$, which is bounded polynomially in $h$, and which has support for the first variable in a compact set independent of $h$ and $\omega'$.

Similarly, $a_k$ may be put in the form $a_k(\omega,\omega';h) = \int_{\mathbb{R}^d} e^{i \langle\omega,x\rangle/h} f_k(x,\omega';h) dx$, where $f_k$ is a function which is smooth in $x$ and $\omega'$, which has support for the first variable in a compact set independent of $h$ and $\omega'$, and which is bounded polynomially in $h$.

We have therefore
\begin{equation*}
\begin{aligned}
\langle \phi_\ell^m, (S_h^k-Id) \phi_\ell^m \rangle& = \int_{\mathbb{S}^{d-1}} d\omega \int_{\mathbb{S}^{d-1}} d \omega' \phi^m_\ell(\omega) \phi^m_\ell(\omega') a_k(\omega,\omega',h)\\
&= \int_{\mathbb{S}^{d-1}}d\omega' \phi^m_\ell(\omega') \int_{\mathbb{R}^d} dx f_k(x,\omega';h)\int_{\mathbb{S}^{d-1}} d\omega e^{i \langle\omega,x\rangle/h}\phi^m_\ell(\omega).
\end{aligned}
\end{equation*}

The last integral is bounded by $O \Big{(} \frac{R}{h \ell^2}\Big{)}^\infty$ thanks to Lemma \ref{harmononstat}. Therefore, since $f_k$ has support for the first variable in a compact set independent of $h$ and $\omega'$, and is bounded polynomially in $h$, we get that 
$\langle \phi_\ell, (S_h^k-Id) \phi_\ell \rangle= O \Big{(} \frac{R}{h \ell^2}\Big{)}^\infty$. We may then sum this estimate over $\ell \geq R'/h^2$ to get $$\sum_{\ell \geq R'/h^2} \sum_{m=1}^{d_\ell}
\langle \phi^m_\ell, (S_h^k-Id) \phi^m_\ell \rangle=O(h^\infty),$$
which concludes the proof of the lemma.
\end{proof}

\paragraph{Putting it all together}
Thanks to equation (\ref{partition}), we have
\begin{equation}\label{litres}
\begin{aligned}
\Tr \big{(}S_{h}^k-Id\big{)} &=\Tr \big{(} (S_{h}^k-Id)(Id-Op_h(\psi_\varepsilon^1)\big{)}\big{)}\\
&+\Tr \big{(}(S_{h}^k-Id)Op_h(\psi_\varepsilon^2(1-\chi_\epsilon^k))\big{)}\\
&+\Tr \big{(}(S_{h}^k-Id)Op_h\big{(}\psi_\varepsilon^2\chi_\varepsilon^k+(\psi_\varepsilon^1- \psi_\varepsilon^2)\big{)}\big{)}+O(h^{2-d}).
\end{aligned}
\end{equation}

To bound the last term, we use that
\begin{equation}\label{ptisup}
\begin{aligned}
\big{|}\Tr \big{(}(S_{h}^k-Id)Op_h\big{(}\psi_\varepsilon^2\chi_\varepsilon^k+(\psi_\varepsilon^1- \psi_\varepsilon^2)\big{)}\big{)}\big{|}&\leq
\|(S_{h}^k-Id)\|_{{L^2(\mathbb{S}^{d-1})\rightarrow L^2(\mathbb{S}^{d-1})}}\\
&\times\big{\|} Op_h\big{(}\psi_\varepsilon^2\chi_\varepsilon^k+(\psi_\varepsilon^1- \psi_\varepsilon^2)\big{)}\big{\|}_{\mathcal{L}^1}+ O(h^{2-d}) \\
&\leq h^{-(d-1)} r'_\varepsilon +O(h^{2-d}),
\end{aligned}
\end{equation}
where $r'_\varepsilon$ is independent of $h$, and is a $O(\varepsilon)$.

Thanks to (\ref{inside}), (\ref{ptisup}) and to Lemma \ref{Out}, equation (\ref{litres}) becomes
\begin{equation*}
h^{(d-1)}\Tr \big{(}S_{h}^k-Id\big{)}= \frac{\Vol(\mathcal{I})}{(2\pi)^{d-1}} +r_\varepsilon + r'_\varepsilon + O(h).
\end{equation*}

Since this is true for any $\varepsilon \bel 0$, we obtain the statement of Proposition \ref{chimp}.
\end{proof}

As a corollary to Proposition \ref{chimp}, we obtain the result for all trigonometric polynomials vanishing at 1, that is, for any function $p$ on $\mathbb{S}^1$ of the form $p(z)= \sum_{-N}^N a_k z^k$ for some coefficients $a_k\in \mathbb{C}$ with $a_0=0$.

\begin{corolaire}\label{brebis}
Suppose that Hypotheses \ref{Trap} and \ref{Div} are satisfied.
Let $p$ be a trigonometric polynomial vanishing at 1. Then we have
$$ \Tr \big{(} p(S_h)\big{)}= \frac{\Vol (\mathcal{I})}{(2\pi h)^{d-1}} \frac{1}{2\pi}
\oint_{\mathbb{S}^1} p(e^{i\theta}) d\theta + o(h^{-(d-1)}).$$
\end{corolaire}
\begin{proof}
Every trigonometric polynomial vanishing at 1 may be written as a linear combination of
polynomials of the form $p(z)=z^k-1$, with $k\in \mathbb{Z}$, for which we
have proved the result in Proposition \ref{chimp}.
\end{proof}
\section{Proof of Theorem \ref{theo}}\label{fin}
Let us define, for any $\alpha>0$,
\begin{equation*}
C^0_{\alpha}(\mathbb{S}^1)= \{f\in C^0(\mathbb{S}^1;\mathbb{C}); f(z)\big{|}\log |z-1|\big{|}^{\alpha} \text{ is continuous } \}.
\end{equation*}

\begin{equation*}
\|f\|_{\alpha} = \sup_{|z|=1, z\neq 1} \big{|}\log |z-1|\big{|}^{\alpha} |f(z)|~~ \text{ for } f\in C^0_{\alpha}(\mathbb{S}^1).
\end{equation*}

Note that $C^0_{\alpha}\subset C^0_{\alpha'}$ if $\alpha>\alpha'$.
We will now prove the following theorem, which is a slightly refined version of Theorem \ref{theo}.

\begin{theoreme} \label{theo2}
Suppose that the manifold $(X,g)$ and the potential $V$ are such that the Hypotheses \ref{Trap} and \ref{Div} are satisfied. Let $\alpha>d$ and 
let $f\in C^0_{\alpha}(\mathbb{S}^1)$. Then we have
\begin{equation*}
\lim\limits_{h\rightarrow 0}\langle \mu_{h},f \rangle = \frac{\Vol(\mathcal{I})}{2\pi}\int_0^{2\pi} f(e^{i\theta})d\theta .
\end{equation*}
\end{theoreme}
Before writing the proof, let us state two technical lemmas. Recall that we denote the eigenvalues of $S_h$ by $e^{i\beta_{n,h}}$. We shall from now on take the convention that $|e^{i\beta_{h,n}}-1|\geq |e^{i\beta_{h,n+1}}-1|$.

For any $L\geq1$, we shall denote by $N_{L,h}$ the number of $n\in \mathbb{N}$ such that $|e^{i\beta_{h,n}}-1|\geq e^{-L/h}$.

\begin{lemme}\label{paranoid2}
There exists $C_0>0$ such that for any $L\geq1$ and $0<h<1$, we have $N_{L,h}\leq C_0 \Big{(}\frac{L}{h}\Big{)}^{d-1}$
\end{lemme}
\begin{proof}
Thanks to equation (2.3) in \cite{christiansen2015sharp}  (which relies on the methods developed in \cite{zworski1989sharp}), we have that there exists $C>0$ independent of $h$ and $n$ such that
\begin{equation}\label{lieneigensingular}
|e^{i\beta_{h,n}}-1|\leq \frac{C}{h^d} \exp \Big{(} \frac{C}{h}- \frac{n^{1/(d-1)}}{C}\Big{)}.
\end{equation}

In particular, we have that for any $N\geq 1$,
\begin{equation*}
\begin{aligned}
\prod_{n=1}^N |e^{i\beta_{h,n}}-1| &\leq \Big{(}\frac{C}{h^d}\Big{)}^N \exp \Big{(} \frac{N C}{h}- \frac{1}{C} \sum_{n=1}^N n^{1/(d-1)}\Big{)}\\
&\leq \Big{(}\frac{C}{h^d}\Big{)}^N \exp \Big{(} \frac{N C}{h}- C' N^{d/(d-1)}\Big{)},
\end{aligned}
\end{equation*}
for some $C'>0$ independent of $h,N$.

Therefore, we have that
\begin{equation*}
\begin{aligned}
e^{-\frac{LN_{L,h}}{h}}&\leq \prod_{n=1}^{N_{L,h}} |e^{i\beta_{h,n}}-1|\\
&\leq \Big{(}\frac{C}{h^d}\Big{)}^{N_{L,h}} \exp \Big{(} \frac{{N_{L,h}} C}{h}- C' {N_{L,h}}^{d/(d-1)}\Big{)}.
\end{aligned}
\end{equation*}

By taking logarithms, we get
\begin{equation*}
\begin{aligned}
-\frac{LN_{L,h}}{h}
&\leq {N_{L,h}}\log\Big{(}\frac{C}{h^d}\Big{)} +  \frac{{N_{L,h}} C}{h}- C' {N_{L,h}}^{d/(d-1)}.
\end{aligned}
\end{equation*}
The first term in the right hand side is negligible, so we get, by possibly changing slightly the constant $C'$,
\begin{equation*}
\begin{aligned}
C' {N_{L,h}}^{d/(d-1)}
&\leq \frac{{N_{L,h}} (C+L)}{h}.
\end{aligned}
\end{equation*}
Therefore, $N_{L,h} \leq \Big{(}\frac{C+L}{C'h}\Big{)}^{d-1}\leq C_0 (L/h)^{d-1}$ for some $C_0>0$ large enough, but independent of $L$ and $h$, which concludes the proof of the lemma.
\end{proof}
\begin{lemme}\label{chevre}
For any $\alpha>d$, there exists $C_\alpha\bel 0$ such that for any $f\in C^0_{\alpha}(\mathbb{S}^1)$, we have
\begin{equation*}
|\langle\mu_{h},f\rangle| \leq C \|f\|_{\alpha}
\end{equation*}
\end{lemme}
\begin{proof}
We have
\begin{equation}\label{mouto}
\begin{aligned}
|\langle\mu_{h},f \rangle| &= (2\pi h)^{d-1} \Big{|}\sum_{n\in
\mathbb{N}} f(e^{i\beta_{h,n}})\Big{|}\\
&\leq (2\pi h)^{d-1} \sum_{|e^{i\beta_{h,n}}-1|\geq e^{-1/h}} |f(e^{i\beta_{h,n}})| + (2\pi h)^{d-1} \sum_{|e^{i\beta_{h,n}}-1|< e^{-1/h}} |f(e^{i\beta_{h,n}})|.
\end{aligned}
\end{equation}

Let us consider the first sum. By Lemma \ref{paranoid2}, it has at most $C_0 h^{-(d-1)}$ terms. Hence, it is bounded by
\begin{equation}
(2\pi h)^{d-1} \sum_{|e^{i\beta_{h,n}}-1|\geq e^{-1/h}} |f(e^{i\beta_{h,n}})|\leq (2\pi h)^{d-1} C_0 h^{-{(d-1)} }\|f\|_{C^0}\leq C \|f\|_{\alpha},
\end{equation}
for some $C>0$.
Let us now consider the second term in (\ref{mouto}). For each $k\geq 1$, we denote by $\sigma_{k,h}$ the set of $n\in \mathbb{N}$ such that $e^{-(k+1)/h}\leq |e^{i\beta_{h,n}}-1|< e^{-k/h}$. By Lemma \ref{paranoid2}, $\sigma_{k,h}$ contains at most $C_0 \Big{(} \frac{k+1}{h}\Big{)}^{d-1}$ elements. On the other hand,
for each $n\in \sigma_{k,h}$, we have 
$$|f(e^{i\beta_{h,n}})| \leq \|f\|_\alpha \big{|}\log (e^{-k/h})\big{|}^{-\alpha}=\frac{h^{\alpha}\|f\|_\alpha}{k^{\alpha}} .$$
 Therefore, we have
\begin{equation*}
\begin{aligned}
(2\pi h)^{d-1}\sum_{|e^{i\beta_{h,n}}-1|< e^{-1/h}} |f(e^{i\beta_{h,n}})| 
&= (2\pi h)^{d-1}\sum_{k=1}^{+\infty} \sum_{n\in \sigma_{k,h}} |f(e^{i\beta_{h,n}})| \\
&\leq  (2\pi h)^{d-1}\sum_{k=1}^{+\infty} C_0 \Big{(} \frac{k+1}{h}\Big{)}^{d-1} \frac{h^{\alpha}\|f\|_\alpha}{k^{\alpha}}\\
&\leq C h^\alpha \|f\|_\alpha,
\end{aligned}
\end{equation*}
for some $C$ independent of $h$.
 This concludes the proof of the lemma.
\end{proof}
\begin{proof}[Proof of Theorem \ref{theo2}]
We have proved the result for all trigonometric polynomials vanishing at 1 in Corollary \ref{brebis}. 
Let $\alpha>\alpha'>d$, and let $f\in C^0_\alpha\subset C^0_{\alpha'}$. Let us show that $f$ can be approximated by trigonometric polynomials vanishing at 1 in the $C^0_{\alpha'}$ norm, which will conclude the proof of the theorem thanks to Lemma \ref{chevre}.

Since $f (z) \big{(}1+\big{|}\log |z-1|\big{|}^{2\alpha}\big{)}^{1/2}$ is continuous, we may find a sequence $P_n$ of polynomials such that $$\big{\|}P_n - f (z) \big{(}1+\big{|}\log |z-1|\big{|}^{2\alpha}\big{)}^{1/2}\big{\|}_{C^0}\leq 1/n.$$ 
Since $f(0)=0$, we may suppose that $P_n(1)=0$. We may also suppose that $P_n'(1)=0$ (for a proof of this fact, see for example \cite[Theorem 8, \S 6]{duren2012invitation}).

Since the function $\big{|}\log |z-1|\big{|}^{\alpha'}\big{(}1+\big{|}\log |z-1|\big{|}^{2\alpha}\big{)}^{-1/2}$ is continuous, we have that 
\begin{equation}\label{approche}
\big{\|}P_n \big{|}\log |z-1|\big{|}^{\alpha'}\big{(}1+\big{|}\log |z-1|\big{|}^{2\alpha}\big{)}^{-1/2} - f (z) \big{|}\log |z-1|\big{|}^{\alpha'}\big{\|}_{C^0}\leq C/n.
\end{equation}

Now, since $P_n(1)=P_n'(1)=0$, the function $P_n /\big{(}(z-1)\big{(}1+\big{|}\log |z-1|\big{|}^{2\alpha}\big{)}^{1/2}\big{)} $ is continuous, and we may find a polynomial $Q_n$ such that 
$$ \Big{\|} \frac{P_n}{(z-1)\big{(}1+\big{|}\log |z-1|\big{|}^{2\alpha}\big{)}^{1/2}} - Q_n\Big{\|}_{C^0}\leq 1/n$$
Since the function $ (z-1)\big{|}\log |z-1|\big{|}^{\alpha'}$ is continuous, we obtain that
\begin{equation}\label{approche2} \Big{\|} P_n\big{|}\log |z-1|\big{|}^{\alpha'}\big{(}1+\big{|}\log |z-1|\big{|}^{2\alpha}\big{)}^{-1/2} - Q_n(z-1)\big{|}\log |z-1|\big{|}^{\alpha'}\Big{\|}_{C^0}\leq C'/n
\end{equation}

Combining (\ref{approche}) and (\ref{approche2}), we obtain that $f$ can be approached by $(z-1)Q_n$ in the $C^0_{\alpha'}$ norm.
This concludes the proof of Theorem \ref{theo2}.

\end{proof}
\bibliographystyle{alpha}
\bibliography{references}
\end{document}